\documentclass[11pt]{article}
\usepackage{amsmath,amssymb,amsthm}
\usepackage{graphicx}
\usepackage{booktabs}
\usepackage[margin=1.15in]{geometry}
\usepackage[colorlinks=true,linkcolor=blue,citecolor=blue,urlcolor=blue]{hyperref}

\newtheorem{lemma}{Lemma}
\newtheorem{corollary}{Corollary}

\title{A Probability Model for Pentagonal Prism Dice Rolls}

\author{Paul R.\ Hurst\thanks{Mathematics, BYU--Hawaii, 55-220 Kulanui
Street, Laie, Hawaii 96762-1294. \texttt{hurstp@byuh.edu}}
\and J.\ Naleo Hyde\thanks{BYU--Hawaii.}}

\date{October 2026}

\begin{document}
\maketitle

\begin{abstract}
Pentagonal prism dice generally have different probabilities of resting on
a pentagonal base and a rectangular side. We derive the centroid solid
angle of a right regular pentagonal prism and study its use as a geometric
baseline for resting probabilities. Nine sizes of approximately equal-volume
PLA dice were printed, and two commercial designs were also tested. A
machine produced 41,068 recorded die outcomes on a hard laminate surface.
The unmodified centroid solid angle disagrees substantially with the
observed proportions. An empirical correction, motivated by the energy
barriers between resting aspects, gives a two-parameter probability model.
We report the fitting criterion, calibrate goodness of fit by a parametric
bootstrap, and provide a likelihood-based sensitivity analysis. The model
is consistent with the original data over the tested range of aspect
ratios under the stated sampling assumptions. Frozen records and
reproducible calculations accompany this version.

\medskip
\noindent\textbf{Keywords:} centroid solid angle, dice probability model,
pentagonal prism dice, unfair dice

\end{abstract}

\noindent\textit{Version note.} This revision preserves the original
41,068 outcomes, clarifies their recording and measurement provenance,
and corrects the goodness-of-fit calibration. Later image batches are not
pooled with this dataset.

\section{Introduction}\label{sec1}

A geometric baseline for a convex polyhedral die of uniform density is
obtained from its centroid solid angles. Place the center of mass at the
origin and project each face radially onto a unit sphere. The area of the
resulting spherical polygon is the solid angle subtended by that face.
Dividing by $4\pi$ gives the fraction of uniformly distributed radial
directions that intersect the face. Figure~\ref{fig:spheres} illustrates
this construction for pentagonal prism dice.

Using these fractions as resting probabilities is an idealized geometric
estimator, rather than a general law of impact dynamics. Related geometric
and quasi-static estimators appear in the automated parts-handling
literature \cite{boo,wieg}.

\begin{figure}[!ht]
\centering
\includegraphics[width=0.85\linewidth]{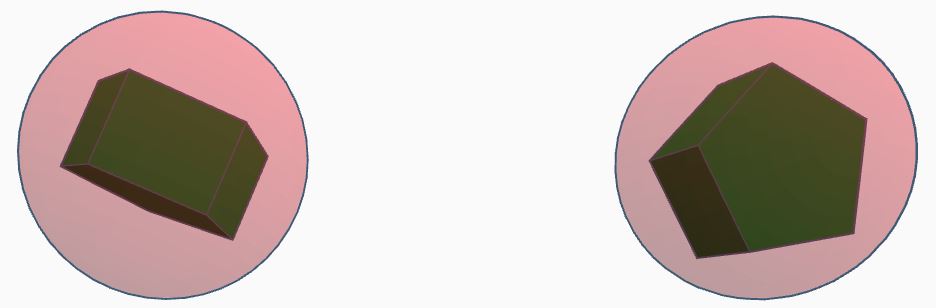}
\caption{Radial sectors of the unit sphere corresponding to a
pentagonal base (left) and a rectangular side (right), illustrating the
centroid-solid-angle construction.}
\label{fig:spheres}
\end{figure}

However, if the experiment consists of dropping the die onto a hard surface
where it can bounce, the centroid solid angle proportions typically do not
give the correct probabilities. \cite{hur} gives insight into why this is
the case. One physical motivation for the discrepancy is the energy barrier window,
discussed in \cite{olm} and is implicit in \cite{boo} and
\cite{boo2}. Related studies of the natural resting aspects of parts in
automated manufacturing, and of stable-pose statistics of polyhedra, include
\cite{ngo0,ngo,chu,wieg}; see also the recent Morse-theoretic treatment of
quasi-static resting probabilities and inverse dice design in \cite{bak}. To understand this, consider a pentagonal prism
whose height is significantly larger than its diameter. Define the
pentagonal-shaped faces as bases and the rectangular-shaped faces as sides.
As the die bounces on a hard surface, the kinetic energy level decreases
until the die reaches rest on either a base or a side. As the energy level
is decreasing, there is a window where the die has enough energy to
transition from being in a base orientation to a side orientation, but not
vice versa. Thus, the probability of getting a base will be less than the
bases' centroid solid angle proportion, while the probability of getting a
side will be greater than the sides' centroid solid angle proportion.

\cite{hur} gave a probability model for square prism dice rolls by adding a
power component, $p=2.427$, to the centroid solid angle model that adjusts
for the energy barrier window. Our probability model for pentagonal prisms
has this component plus an additional component, $a=1.46$, to account for
the different shapes of the base and side faces. We define the height $h$ of
a die as the distance between the two pentagonal faces. The radius, $r$, is
the distance from the center to any vertex of a pentagonal face. See
Figure~\ref{fig:hr}. When $h = ar \approx 1.46r$, our model and the centroid
solid angle proportions match up.

\section{Centroid Solid Angles}\label{sec2}

In order to find the Centroid Solid Angles (CSA), we will start by examining
the top face of a prismatic pentagonal die. Place the die in the
$xyz$-coordinate system with the center of mass as the origin ($O$) and the
top face parallel to the $xy$-plane, as in Figures~\ref{fig:axis1}
and~\ref{fig:axis2}.

\begin{figure}[!ht]
\centering
\includegraphics[width=.55\linewidth]{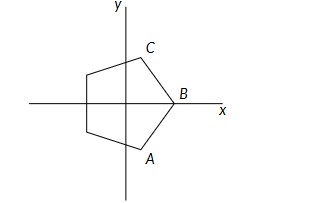}
\caption{\textit{A right regular pentagonal prism placed in the
$xyz$-coordinate system, with its center of mass at the origin $O$ and its
top face parallel to the $xy$-plane.}}
\label{fig:axis1}
\end{figure}

\begin{figure}[!ht]
\centering
\includegraphics[width=.55\linewidth]{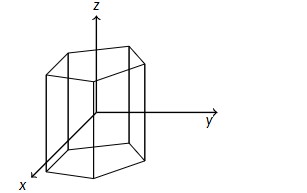}
\caption{\textit{The top face of the same die, showing three consecutive
vertices $A$, $B$ and $C$ and the height $h$ and circumradius $r$ used
throughout.}}
\label{fig:axis2}
\end{figure}

\begin{lemma}
\label{alpha}
Let $h$ be the height of a right regular pentagonal prism of uniform
density with center of mass $O$, and let $r$ be the distance in the top
face from its center to any vertex. Let $A$, $B$, and $C$ be three
consecutive vertices of the top face, positioned as in
Figures~\ref{fig:axis1} and~\ref{fig:axis2}.
The interior dihedral angle $\alpha$ along the ray $OB$ between the planes
$ABO$ and $BCO$ is
\begin{align*}
   \alpha &= \pi - \cos^{-1}\left(
    \frac{
   \sin^2\left(\frac{2\pi}{5}\right) + 2\cos\left(\frac{2\pi}{5}\right) - \cos^2\left(\frac{2\pi}{5}\right) - 1 + 4\left(\frac{r}{h}\right)^2\sin^2\left(\frac{2\pi}{5}\right)
   }
   {
   \sin^2\left(\frac{2\pi}{5}\right) + 1 - 2\cos\left(\frac{2\pi}{5}\right) + \cos^2\left(\frac{2\pi}{5}\right) + 4\left(\frac{r}{h}\right)^2\sin^2\left(\frac{2\pi}{5}\right)
   }\right)
   \end{align*}
\end{lemma}

\begin{proof}
Note the coordinates of the vertices $A$, $B$, and $C$ are:
\begin{align*}
\vec{A} &= \left( r\cos\left(\tfrac{2\pi}{5}\right), r\sin\left(\tfrac{-2\pi}{5}\right), \tfrac{h}{2}\right)\\
\vec{B} &= \left(r, 0, \tfrac{h}{2}\right)\\
\vec{C} &= \left( r\cos\left(\tfrac{2\pi}{5}\right), r\sin\left(\tfrac{2\pi}{5}\right), \tfrac{h}{2}\right)
\end{align*}

Next, we need to calculate the normal vectors to the planes $ABO$ and $BCO$,
$\vec{n}_1 = \vec{A} \times \vec{B}$ and $\vec{n}_2 = \vec{B} \times
\vec{C}$. These are found to be:
\begin{align*}
    \vec{n}_1 &= \vec{A} \times \vec{B} = \left(\frac{-hr\sin\left(\frac{2\pi}{5}\right)}{2}, \frac{hr - hr\cos\left(\frac{2\pi}{5}\right)}{2}, r^2\sin\left(\frac{2\pi}{5}\right)\right)\\
    \vec{n}_2 &= \vec{B} \times \vec{C} = \left(\frac{-hr\sin\left(\frac{2\pi}{5}\right)}{2}, \frac{ hr\cos\left(\frac{2\pi}{5} \right) - hr}{2}, r^2\sin\left(\frac{2\pi}{5}\right)\right)
\end{align*}

Next, we need to calculate the angle between these two normal vectors, and
to do so, we use the formula:
\begin{align*}
    \cos^{-1}\left(\frac{\vec{n}_1 \cdot \vec{n}_2}{||\vec{n}_1|| || \vec{n}_2||}\right)
\end{align*}

Which, calculating, gives us:
\small
\begin{align*}
   \vec{n}_1 \cdot \vec{n}_2 &=  \frac{h^2r^2\sin^2\left(\frac{2\pi}{5}\right)}{4} + \frac{2r^2h^2\cos\left(\frac{2\pi}{5}\right) - h^2r^2\cos^2\left(\frac{2\pi}{5}\right) - h^2r^2}{4} + r^4\sin^2\left(\frac{2\pi}{5}\right)\\
   &= \frac{r^2\left(h^2\sin^2\left(\frac{2\pi}{5}\right) + 2h^2\cos\left(\frac{2\pi}{5}\right) - h^2\cos^2\left(\frac{2\pi}{5}\right) - h^2 + 4r^2\sin^2\left(\frac{2\pi}{5}\right)\right)}{4}
\end{align*}
\normalsize

Continuing, we have the magnitude of $\vec{n}_1$
\small
\begin{align*}
    ||\vec{n}_1|| &=
    \sqrt{\frac{h^2r^2\sin^2\left(\frac{2\pi}{5}\right)}{4} + \frac{h^2r^2 - 2h^2r^2\cos\left(\frac{2\pi}{5}\right) + h^2r^2\cos^2\left(\frac{2\pi}{5}\right)}{4} + r^4\sin^2\left(\frac{2\pi}{5}\right)}
\end{align*}
\normalsize
Note that the magnitude of $||\vec{n}_1|| = ||\vec{n_2}||$.

Thus, it follows that:
\small
\begin{align*}
   \frac{\vec{n}_1 \cdot \vec{n}_2}{||\vec{n}_1||||\vec{n}_2||} &=
   \frac{
   \frac{r^2\left(h^2\sin^2\left(\frac{2\pi}{5}\right) + 2h^2\cos\left(\frac{2\pi}{5}\right) - h^2\cos^2\left(\frac{2\pi}{5}\right) - h^2 + 4r^2\sin^2\left(\frac{2\pi}{5}\right)\right)}{4}
   }
   {
   \frac{h^2r^2\sin^2\left(\frac{2\pi}{5}\right)}{4} + \frac{h^2r^2 - 2h^2r^2\cos\left(\frac{2\pi}{5}\right) + h^2r^2\cos^2\left(\frac{2\pi}{5}\right)}{4} + r^4\sin^2\left(\frac{2\pi}{5}\right)
   }\\
   &=
   \frac{
   r^2\left(h^2\sin^2\left(\frac{2\pi}{5}\right) + 2h^2\cos\left(\frac{2\pi}{5}\right) - h^2\cos^2\left(\frac{2\pi}{5}\right) - h^2 + 4r^2\sin^2\left(\frac{2\pi}{5}\right)\right)
   }
   {
   h^2r^2\sin^2\left(\frac{2\pi}{5}\right) + h^2r^2 - 2h^2r^2\cos\left(\frac{2\pi}{5}\right) + h^2r^2\cos^2\left(\frac{2\pi}{5}\right) + 4r^4\sin^2\left(\frac{2\pi}{5}\right)
   }\\
   &=
   \frac{
   h^2\sin^2\left(\frac{2\pi}{5}\right) + 2h^2\cos\left(\frac{2\pi}{5}\right) - h^2\cos^2\left(\frac{2\pi}{5}\right) - h^2 + 4r^2\sin^2\left(\frac{2\pi}{5}\right)
   }
   {
   h^2\sin^2\left(\frac{2\pi}{5}\right) + h^2 - 2h^2\cos\left(\frac{2\pi}{5}\right) + h^2\cos^2\left(\frac{2\pi}{5}\right) + 4r^2\sin^2\left(\frac{2\pi}{5}\right)
   }\\
   &=
   \frac{
   \sin^2\left(\frac{2\pi}{5}\right) + 2\cos\left(\frac{2\pi}{5}\right) - \cos^2\left(\frac{2\pi}{5}\right) - 1 + 4\left(\frac{r}{h}\right)^2\sin^2\left(\frac{2\pi}{5}\right)
   }
   {
   \sin^2\left(\frac{2\pi}{5}\right) + 1 - 2\cos\left(\frac{2\pi}{5}\right) + \cos^2\left(\frac{2\pi}{5}\right) + 4\left(\frac{r}{h}\right)^2\sin^2\left(\frac{2\pi}{5}\right)
   }
\end{align*}
\normalsize

The angle between the normals to the planes is the inverse cosine of the
last expression. With the orientation conventions above, $\vec{n}_1$ and
$\vec{n}_2$ both point into the wedge between the two planes. The angle between two
consistently oriented inward normals is supplementary to the interior
dihedral angle $\alpha$ between the planes $ABO$ and $BCO$. Thus:
\begin{align*}
   \alpha &= \pi - \cos^{-1}\left(
    \frac{
   \sin^2\left(\frac{2\pi}{5}\right) + 2\cos\left(\frac{2\pi}{5}\right) - \cos^2\left(\frac{2\pi}{5}\right) - 1 + 4\left(\frac{r}{h}\right)^2\sin^2\left(\frac{2\pi}{5}\right)
   }
   {
   \sin^2\left(\frac{2\pi}{5}\right) + 1 - 2\cos\left(\frac{2\pi}{5}\right) + \cos^2\left(\frac{2\pi}{5}\right) + 4\left(\frac{r}{h}\right)^2\sin^2\left(\frac{2\pi}{5}\right)
   }\right) \qedhere
   \end{align*}
\end{proof}

\begin{corollary}\label{cor:csa}
The proportion of the unit sphere subtended, from the center of mass, by
the two pentagonal bases together is
\[
\frac{5\alpha - 3\pi}{2\pi},
\]
with $\alpha$ as in Lemma~\ref{alpha}.
\end{corollary}

\begin{proof}
Projecting a base onto the enclosing unit sphere yields a spherical
pentagon whose sides are great-circle arcs; the angle of the spherical
pentagon at each vertex equals the dihedral angle between the two
corresponding planes through $O$, which by symmetry is $\alpha$ at all five
vertices. Triangulate the spherical pentagon into three spherical triangles. Summing
Girard's triangle-area formula gives the sum of the pentagon's angles
minus $3\pi$; here
this gives $5\alpha - 3\pi$. The die has two bases, so the subtended area
is $2(5\alpha - 3\pi)$ out of the total sphere area $4\pi$, giving the
stated proportion.
\end{proof}

Two limiting cases serve as checks. As $h/r \to \infty$ the argument of
$\cos^{-1}$ in Lemma~\ref{alpha} tends to $\cos\left(\frac{2\pi}{5}\right)$,
so $\alpha \to \frac{3\pi}{5}$ (the interior angle of a planar regular
pentagon) and the proportion tends to $0$; as $h/r \to 0$, $\alpha \to \pi$
and the proportion tends to $1$, as expected.

A shorter equivalent form will be useful below. With $u=r/h$, applying the
half-angle identity to the cosine in Lemma~\ref{alpha} gives
\[
 \alpha=\pi-2\arctan\!\left(\frac{\tan(\pi/5)}{\sqrt{1+4u^2}}\right),
 \qquad
 G(u)=1-\frac{5}{\pi}\arctan\!\left(
          \frac{\tan(\pi/5)}{\sqrt{1+4u^2}}\right).
\]
Here $G$ denotes the combined CSA proportion of the two bases.

\section{Data Generation}\label{sec3}

Nine different sizes of pentagonal prism dice of polylactic acid (PLA) were
printed. They all have roughly the same volumes, but with varying heights
and widths. The printer was set to 100\% infill; uniform density is an approximation
used in the geometric model.
The dice were lightly sanded. In addition, two different brands of
commercial pentagonal prism dice were obtained; see
Figure~\ref{fig:alldice}.

\begin{figure}[!ht]
\centering
\includegraphics[width=0.9\linewidth]{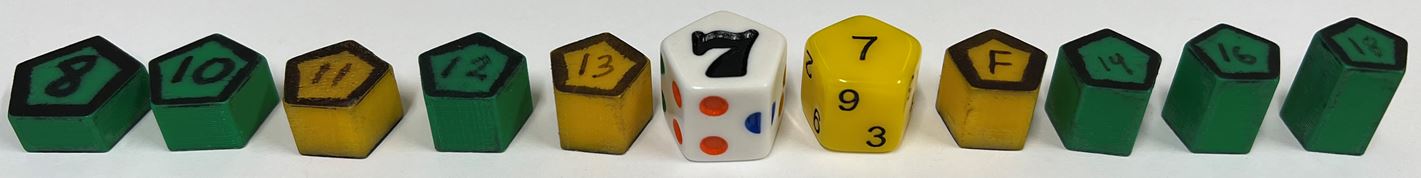}
\caption{\textit{All dice sizes, arranged from smallest to largest
height-to-radius ratio.}}
\label{fig:alldice}
\end{figure}

The nine printed sizes comprise dice 8--18, labeled by their nominal
height in millimetres, together with die F, which was designed to approach
equal probabilities for all seven faces. DS and AD denote the commercial
designs. Our recorded response is base versus side; it does not separately
test equality among the five individual side faces of the printed dice.

Heights were measured with calipers by measuring stacks of four copies
and dividing. The circumradius was computed from vertex-to-opposite-edge
width $w$ as $r=w/(1+\cos(\pi/5))$. The analysis workbook applies width
adjustments of $+0.020$\,mm for printed dice, $+0.530$\,mm for DS, and
$+0.470$\,mm for AD, to approximate sanding and rounded-corner effects.
These are the adjustments actually used for Table~\ref{tab:dims}.

Twenty individual width readings are available for eight printed sizes
and for each commercial design; for F the archived analysis provides a
mean width and a height. Printed width spreads are typically several
hundredths of a millimetre. The sample standard deviations for DS and AD
are $0.159$\,mm and $0.072$\,mm, respectively. The printed volumes computed
from these dimensions range from about 2,091 to 2,231\,mm$^3$; DS and AD
are larger, at about 4,698 and 3,654\,mm$^3$. The analysis treats the
corrected dimensions as fixed; measurement and geometric-model
uncertainty are not propagated into the statistical errors below.
Table~\ref{tab:dims} lists the dimensions and recorded counts.

\begin{table}[!ht]
  \centering
  \caption{Measured dimensions and roll counts. $k$ is the number of rolls
  ending on a base (either pentagonal face).}
    \begin{tabular}{lrrrr}
    \toprule
    Label & $h$ (mm) & $r$ (mm) & $k$ (base) & $n$ \\
    \midrule
    8  & 7.930  & 10.877 & 3985 & 4509 \\
    10 & 9.908  & 9.709  & 3701 & 5451 \\
    11 & 11.140 & 8.886  & 1491 & 2879 \\
    12 & 11.828 & 8.843  & 2276 & 5070 \\
    DS & 16.313 & 11.005 & 715  & 2056 \\
    13 & 12.650 & 8.457  & 512  & 1545 \\
    AD & 15.413 & 9.986  & 890  & 2921 \\
    F  & 12.973 & 8.297  & 630  & 2131 \\
    14 & 13.840 & 8.167  & 1049 & 4734 \\
    16 & 15.813 & 7.633  & 601  & 5091 \\
    18 & 17.848 & 7.182  & 280  & 4681 \\
    \midrule
       &        &        &      & 41{,}068 \\
    \bottomrule
    \end{tabular}
  \label{tab:dims}
\end{table}

Using a machine that simulates dropping dice through a commercial dice
tower onto a hard surface, 41,068 die outcomes were recorded. Drops usually
contained three dice, with ten recorded four-dice drops in the printed
series. DS dice, comprising white and black copies, were usually dropped
two at a time. Partial drops were retained.

Each descent sends the dice down a ramp inclined at $30^\circ$ to the
horizontal, onto an opposing ramp at $45^\circ$, and against a third ramp
at approximately $70^\circ$. The dice then fall about one foot onto a
smooth hard laminate surface, the shelf material described in \cite{hur},
where they bounce and come to rest.

The primary counts were reconstructed from the first author's entry
columns. The public workbook contains independent second-author entries
for all recorded drops of six printed sizes, and partial coverage for
11, F, and 18: 550 of 973, 687 of 716, and 375 of 1,580 recorded drops,
respectively. Five disagreements remain between populated entry pairs in
the archived workbook; inspection of the corresponding photographs
supports the primary entries used here. The commercial primary counts
were recorded by the first author. The ancillary files preserve both entry
columns and document this coverage rather than treating missing second
entries as cross-checks.

The ramps are intended to randomize orientation and velocity. Contacts
between dice can occur. We use independent Bernoulli outcomes as a working
sampling model; randomization by the apparatus does not itself establish
independence or eliminate possible differences between copies and sessions.
A preliminary calculation from the recorded drop counts gives small
within-drop correlations for most sizes; the largest positive estimate
among the printed sizes is approximately $0.052$ for die 16. Statistical
uncertainties below are conditional on the Bernoulli model and fixed
geometry.

A \href{https://www.youtube.com/watch?v=OIzivrHaPfY}{video of the rolling
machine} and \href{https://byuh.box.com/s/pu777x3alpztu0txc536g875463u0cqy}{recorded
images} provide experimental context. This version uses the original
recorded dataset. Later die-13 sessions and unentered white-DS images
are excluded pending their separate verification and protocol reconciliation.

\section{Modified Centroid Solid Angle Model}\label{sec4}

Table~\ref{tab:csa} illustrates the probabilities obtained from the data.
Throughout, for a die rolled $n$ times with observed base proportion
$\hat p$ and a model probability $P$, the reported $Z$-score is
$Z = (\hat p - P)/\sqrt{P(1-P)/n}$.

\begin{table}[!ht]
  \centering
  \caption{CSA Z-Scores}
    \begin{tabular}{lrrrr}
    \toprule
    Label & CSA Base & $n$ & Prop.\ of Base & Z-Score \\
    \midrule
    8  & 0.6119 & 4509 & 0.8838 & 37.4678 \\
    10 & 0.4924 & 5451 & 0.679  & 27.5470 \\
    11 & 0.4139 & 2879 & 0.5179 & 11.3308 \\
    12 & 0.3891 & 5070 & 0.4489 & 8.7404 \\
    DS & 0.3502 & 2056 & 0.3478 & $-0.2306$ \\
    13 & 0.3468 & 1545 & 0.3314 & $-1.2691$ \\
    AD & 0.3351 & 2921 & 0.3047 & $-3.4848$ \\
    F  & 0.3304 & 2131 & 0.2956 & $-3.4090$ \\
    14 & 0.3011 & 4734 & 0.2216 & $-11.9324$ \\
    16 & 0.2335 & 5091 & 0.1181 & $-19.4646$ \\
    18 & 0.1803 & 4681 & 0.0598 & $-21.4387$ \\
    \midrule
       &        & 41{,}068 & & \\
    \bottomrule
    \end{tabular}
  \label{tab:csa}
\end{table}

Note that the observed proportions do not match the centroid solid angles:
nine of the eleven dice differ from the CSA prediction by more than three
standard errors, with $Z$-scores as large as $37$. One motivation for this discrepancy is the energy-barrier window
\cite{olm,hur}. Let $c=r\cos(\pi/5)$ be the pentagon's apothem. From a
side-resting aspect, tipping toward another side raises the center of mass
by $r-c$, whereas tipping toward a base raises it by
$\sqrt{c^2+h^2/4}-c$. Multiplication by $mg$ gives the corresponding
energy barriers. The first is smaller than the second when
$h/r>2\sin(\pi/5)\approx1.176$; at equal resting center-of-mass heights,
$h/2=c$, this condition holds. These geometric differences motivate an
empirical correction, but do not by themselves derive its functional form.

Let $h$ be the distance from base to base. Let $r$ be the distance from the
center of the pentagon to a vertex as illustrated in Figure~\ref{fig:hr}.

\begin{figure}[!ht]
\centering
\includegraphics[width=.6\linewidth]{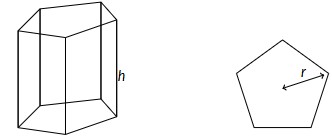}
\caption{\textit{Height ($h$) and Radius ($r$) of the dice}}
\label{fig:hr}
\end{figure}

We modify the CSA proportions by weighting the base contribution by
$h^{-p}$ and the side contribution by $(ar)^{-p}$, then renormalizing.
The dimensionless scale $a$ controls the crossing with CSA, and the power
$p$ controls the change in the relative odds. We use the same empirical
family as the earlier draft, with $(a,p)=(1.46,2.33)$, selected on a
$0.01$ grid by minimizing the largest absolute binomial $Z$-score over
the eleven sizes. The corresponding square-prism exponent in \cite{hur}
is $2.427$; no common exponent across different prism families is assumed.

For plotting, write $u=r/h$ and $x=\ln(au)$, so $x=0$ corresponds to the
crossing $h=ar$. In terms of the CSA proportion derived above,
\[
 g(x)=G(e^x/a),\qquad
 P_B(x)=\frac{g(x)/h^p}{g(x)/h^p+(1-g(x))/(ar)^p}
       =\frac{g(x)}{g(x)+(1-g(x))e^{-px}}.
\]
At fixed physical geometry, $G(u)$ is independent of the fitted scale $a$.
Equivalently, with $\operatorname{logit}(v)=\ln[v/(1-v)]$,
\[
 \operatorname{logit}P_B=\operatorname{logit}G(u)+p\ln u+p\ln a.
\]
This representation also permits a standard binomial-likelihood fit with
known offset $\operatorname{logit}G(u)$.

\begin{table}[!ht]
\centering
\caption{Modified CSA fit, $a=1.46$, $p=2.33$.}
\begin{tabular}{lrrrrr}
\toprule
Label & $\ln(ar/h)$ & $n$ & Observed base & Model base & $Z$ \\
\midrule
8 & 0.694 & 4509 & 0.8838 & 0.8883 & -0.958 \\
10 & 0.358 & 5451 & 0.6790 & 0.6909 & -1.910 \\
11 & 0.152 & 2879 & 0.5179 & 0.5018 & 1.730 \\
12 & 0.088 & 5070 & 0.4489 & 0.4385 & 1.490 \\
DS & -0.015 & 2056 & 0.3478 & 0.3422 & 0.530 \\
13 & -0.024 & 1545 & 0.3314 & 0.3341 & -0.222 \\
AD & -0.056 & 2921 & 0.3047 & 0.3069 & -0.261 \\
F & -0.068 & 2131 & 0.2956 & 0.2961 & -0.045 \\
14 & -0.149 & 4734 & 0.2216 & 0.2334 & -1.924 \\
16 & -0.350 & 5091 & 0.1181 & 0.1188 & -0.159 \\
18 & -0.532 & 4681 & 0.0598 & 0.0599 & -0.017 \\
\midrule
 & & 41{,}068 & & & \\
\bottomrule
\end{tabular}
\label{tab:mod}
\end{table}

The main fit has maximum $|Z|=1.924$ and Pearson sum
$Q=\sum Z^2=13.9075$. Since the parameters were chosen to minimize the
maximum residual, the statement that all eleven residuals are below two
is descriptive rather than an independent goodness-of-fit test.

We calibrate $Q$ by a parametric bootstrap that repeats the same grid
minimax fit. For each of 4,000 replicates, independent binomial counts are
generated at the fitted probabilities with the original sample sizes,
then both parameters are refitted on grids $1.30\leq a\leq1.65$ and
$1.80\leq p\leq2.90$, with spacing $0.01$. No replicate optimum reaches a
grid boundary. Comparing the refitted sums with the observed sum gives a
bootstrap $P$-value of approximately $0.26$ (Monte Carlo standard error
$0.007$). Bootstrap standard deviations of the fitted parameters are
approximately $0.011$ for $a$ and $0.061$ for $p$, conditional on this
sampling model and the measured geometry.

As a sensitivity analysis, maximizing the binomial likelihood instead gives
$a=1.4568$ and $p=2.3058$, with conditional standard errors $0.0077$ and
$0.0432$. Its Pearson sum is $13.2725$; the usual likelihood-fit
chi-square approximation with $11-2=9$ degrees of freedom gives a
$P$-value of approximately $0.15$. Both assessments are consistent with
the two-parameter model over the tested range under the stated assumptions.
They do not establish a physical derivation of the empirical correction.

Refitting the minimax model to the nine printed sizes only gives the same
grid pair $(1.46,2.33)$. For the excluded commercial sizes, the resulting
residuals are $Z=+0.53$ for DS and $Z=-0.26$ for AD. Both lie close to the
crossing, so this check has limited sensitivity to $p$ and does not
validate predictions far from the crossing. Rounded-corner corrections
also introduce geometric uncertainty that this check does not quantify.
At $x=0$, $P_B(0)=g(0)\approx0.356$ exactly by the model's construction.

Figure~\ref{fig:models} is a graph of the CSA ($g$) model and the Modified
CSA ($P_B$) model.

\begin{figure}[!ht]
\centering
\includegraphics[width=0.95\linewidth]{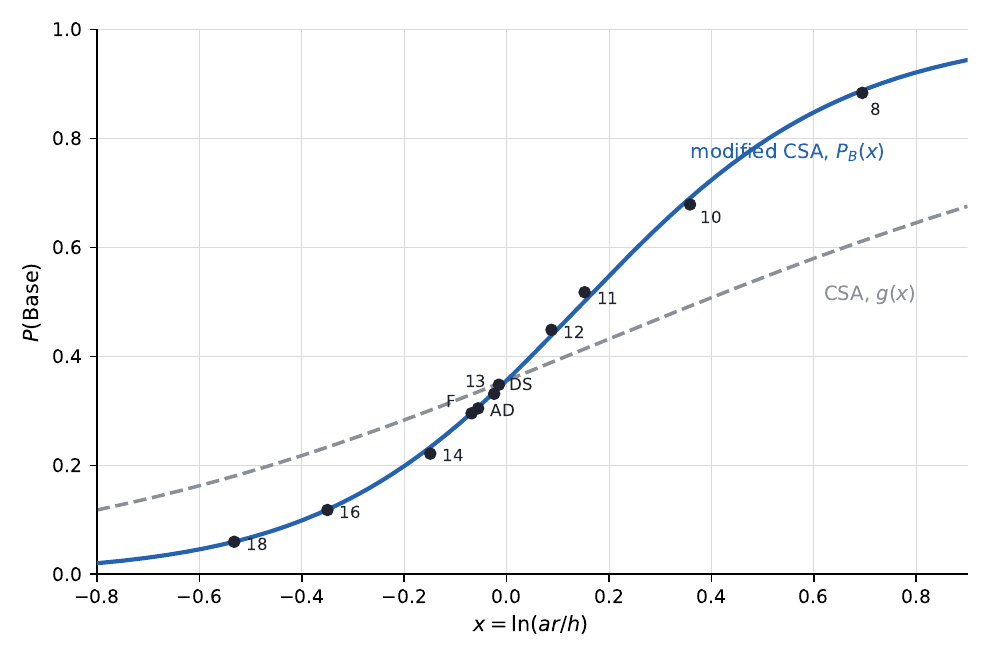}
\caption{CSA and Modified CSA models with the eleven data points.}
\label{fig:models}
\end{figure}

\section{Conclusion and Future Research}\label{sec5}

The modified CSA model describes the original pentagonal-prism
base-versus-side counts substantially better than unmodified CSA. Its
parameters and uncertainty depend on the sampling assumptions, corrected
ideal geometry, and the apparatus used. The experiment does not separately
establish scale invariance, uniformity among physical copies, or equal
probabilities of individual side faces.

Further work will compare triangular, square, pentagonal and hexagonal
prisms with cylinders (thick coins), a problem with a long history
\cite{mos}. Increasing the number of polygon sides approaches a circular
cylinder geometrically, as illustrated in Figure~\ref{fig:thickcoin}.
Comparisons across these families and additional experimental batches
will be developed separately from the present dataset.

\begin{figure}[!ht]
\centering
\includegraphics[width=0.9\linewidth]{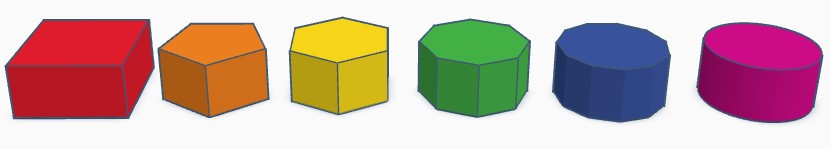}
\caption{\textit{Polygonal prisms ``converging'' to the thick coin}}
\label{fig:thickcoin}
\end{figure}

\section*{Author contributions}
The authors worked together on data collection, analysis, and writing the
paper.

\section*{Use of artificial intelligence}
AI assistants were used for language editing, draft review, numerical
checks, preparation of reproducibility files, and inspection of selected
recorded images during the data audit. The original experimental design,
physical data collection, and empirical model preceded this assistance.
The authors are responsible for the manuscript and its conclusions.

\section*{Data availability}
Table~\ref{tab:dims} gives the primary per-size counts. Ancillary files
supplied with this version contain full-precision dimensions and counts,
a frozen snapshot of the public recording workbook, a CSV preserving the
populated original entry columns, available caliper measurements, and a
Python/NumPy script reproducing the tables, likelihood sensitivity check,
and seeded bootstrap. The measurement file distinguishes individual
readings from the recorded mean used for F. Recording coverage, residual
entry disagreements, and excluded later batches are documented in the
ancillary data notes.

The \href{https://docs.google.com/spreadsheets/d/1vV9ETh5S6nfFtguG9UGl5c2ENJM_JRSRXTcYkJL9epI/edit?usp=sharing}{working
Google Sheet} remains available as a convenience link. It can change;
the archived snapshot and per-size CSV identify the dataset used in this
version. Links to the machine video and images are in Section~\ref{sec3}.

\section*{Financial disclosure}
None reported.

\section*{Conflict of interest}
The authors declare no conflicts of interest.

\end{document}